\documentclass[12pt]{article}
\usepackage{amsfonts}
\usepackage{enumerate}
\usepackage{amsmath}
\usepackage{amsthm}
\usepackage{color}
\usepackage{addfont}

\def\1{{\bf 1}}

\begin{document}

\newtheorem{thm}{Theorem}[section]
\newtheorem{lem}[thm]{Lemma}
\newtheorem{prop}[thm]{Proposition}
\newtheorem{cor}[thm]{Corollary}
\newtheorem{defi}[thm]{Definition}
\newtheorem{remark}[thm]{Remark}
\newtheorem{result}[thm]{Result}
\newtheorem{exm}[thm]{Example}
\title
{\bf Construction of extremal Type II $\mathbb{Z}_{8}$-codes via doubling method}
\author
{Sara Ban (sban@math.uniri.hr), ORCID: 0000-0002-1837-8701   \\
[3pt]
Sanja Rukavina (sanjar@math.uniri.hr), ORCID: 0000-0003-3365-7925 \\
{\it  Faculty of Mathematics, University of Rijeka}\\
{\it Radmile Matej$\check{c}$i\'{c} 2, 51000 Rijeka, Croatia }\\[-15pt]
}
\date{}
\maketitle

\bigskip

\vspace*{0.2cm}

\begin{abstract}
\noindent
Extremal Type II $\mathbb{Z}_{8}$-codes are a class of self-dual $\mathbb{Z}_{8}$-codes with Euclidean weights divisible by $16$ and the largest possible minimum Euclidean weight for a given
length. 
We introduce a doubling method for constructing a Type II $\mathbb{Z}_{2k}$-code of length $n$ from a known  Type II $\mathbb{Z}_{2k}$-code of length $n$. 
Based on this method, we develop an algorithm to construct new extremal Type II $\mathbb{Z}_8$-codes starting from an extremal Type II $\mathbb{Z}_8$-code of type $(\frac{n}{2},0,0)$ with an extremal $\mathbb{Z}_4$-residue code and length $24, 32$ or $40$.\\ 
We construct at least ten new extremal Type II $\mathbb{Z}_8$-codes of length $32$ and type $(15,1,1)$. Extremal Type II $\mathbb{Z}_8$-codes of length $32$ of this type were not known before. Moreover, the binary residue codes of the constructed extremal $\mathbb{Z}_8$-codes are optimal $[32,15]$ binary codes. 
 
\end{abstract}


\section{Introduction}

The discovery of good nonlinear binary codes arising via the Gray map from $\mathbb{Z}_4$-linear codes motivated the study of codes over rings in general (see \cite{1}). Construction of unimodular lattices with large minimum norm has motivated
the construction of new self-dual $\mathbb{Z}_{2k}$-codes with large minimum Euclidean weights (see, for example, \cite{Bannai, Doug}).
Especially, $\mathbb{Z}_{8}$-codes have received attention by many researchers.
For instance,  some
construction methods for self-dual codes over $\mathbb{Z}_{8}$ for arbitrary length greater than $8$ are given in \cite{sup}.

Extremal Type II $\mathbb{Z}_{8}$-codes are a class of self-dual $\mathbb{Z}_{8}$-codes with Euclidean weights divisible by $16$ and the largest possible minimum Euclidean weight for a given
length. For lengths $8$ and $16,$ every Type II $\mathbb{Z}_{8}$-code is extremal (see \cite{Hnovi}).
In \cite{6810}, the previously known results on the existence of extremal Type II $\mathbb{Z}_8$-codes for greater lengths are summarized: there are three such codes of length $24$ (see \cite{Bannai, Georgiou}), five codes of length $32,$ a large number of codes of length $40$ and one such code of length $48$ (see \cite{Chapman, Conway}), up to equivalence. In addition, a new extremal Type II $\mathbb{Z}_8$-code $D_n$ is constructed for each of the lengths $n \in\{24,32,40\}$ (see [\cite{6810}, Subsection 5.2]).

The doubling method for a construction of Type II $\mathbb{Z}_4$-codes is introduced in \cite{Chan}. It was used for the construction of extremal Type II $\mathbb{Z}_4$-codes of length $32$ and $40$ in \cite{MM} and \cite{40doubling}, respectively.
In this paper, we introduce a doubling method for constructing a Type II $\mathbb{Z}_{2k}$-code of length $n$ from a known  Type II $\mathbb{Z}_{2k}$-code of length $n$. We also develop an algorithm that uses this method to construct a Type II $\mathbb{Z}_{2^m}$-code of length $n$ from a known  Type II $\mathbb{Z}_{2^m}$-code of length $n$. Finally, by specifying the method for $m=3$ and $n \in \{24,32,40\},$ we obtain an algorithm that is then used to construct at least $10$ new extremal Type II $\mathbb{Z}_{8}$-codes of length $32$.

The paper is organized as follows. The next section gives definitions and basic properties of codes over $\mathbb{Z}_{2k}$ that will be needed in our work. 
In Section \ref{constr}, we introduce the doubling method to construct new Type II $\mathbb{Z}_{2k}$-codes starting from a Type II $\mathbb{Z}_{2k}$-code. Especially, we consider the codes over $\mathbb{Z}_{2^m}.$ 
Finally, in the last section, we present a method to construct new extremal Type II $\mathbb{Z}_8$-codes starting from an extremal Type II $\mathbb{Z}_8$-code of type $(\frac{n}{2},0,0)$ with an extremal $\mathbb{Z}_4$-residue code and length $24, 32$ or $40$. 
Using this method, we construct $68850$ extremal Type II $\mathbb{Z}_8$-codes of type $(15,1,1)$ and length $32.$  
We give the weight distributions of the corresponding binary residue codes. With respect to these weight distributions, all constructed extremal Type II $\mathbb{Z}_8$-codes are divided into ten classes. For each of them we give a generator matrix in the standard form for one class representative.

\section{Preliminaries}\label{pre}

For terms not defined in this paper and the basic facts of coding theory we refer the reader to \cite{BBF, FEC, vdt}.

Let $\mathbb{Z}_{2k}$ denote the ring of integers modulo $2k.$ A linear code $C$ of length $n$ over $\mathbb{Z}_{2k}$ (i.e., a {\it $\mathbb{Z}_{2k}$-code}) is an additive subgroup  of $\mathbb{Z}_{2k}^n.$  
Two codes over $\mathbb{Z}_{2k}$ are {\it equivalent} if one can be obtained from the other by permuting the coordinates and (if necessary) changing the signs of certain coordinates. Codes differing by only a
permutation of coordinates are called {\it permutation-equivalent}.
  An element of $C$ is called a {\it codeword} of $C.$ A {\it generator matrix} of $C$ is a matrix whose rows generate $C.$

The {\it dual code} $C^\bot$ of $C$ is defined as
$$
C^\bot=\{x\in \mathbb{Z}_{2k}^n\,|\,\left\langle x, y\right\rangle=0\ \text{for all}\ y\in C\},
$$
where $\left\langle x,\\ y\right\rangle=x_1y_1+x_2y_2+\dots+x_ny_n\ (\text{mod}\ 2k)$ for
$x=(x_1,x_2,\dots,x_n)$ and $y=(y_1,y_2,\dots,y_n).$ The code $C$ is {\it self-orthogonal} when $C\subseteq C^\bot$ and {\it self-dual} if $C= C^\bot.$

The {\it Euclidean weight} of a codeword $x=(x_1,x_2,\dots,x_n) \in \mathbb{Z}_{2k}^n$ is 
$$wt_E (x) =\sum_ {i=1}^n \min \{x_i^2, (2k-x_i)^2 \}.$$ 

It holds \begin{equation}\label{et} wt_E(x+y)\equiv wt_E(x)+wt_E(y)+2\left\langle x, y\right\rangle \ (\text{mod}\ 4k)\end{equation}
for all $x, y \in \mathbb{Z}_{2k}^n$ (see \cite{Bannai}).  We denote the number of coordinates $i$ (where $i=0,1,\dots,2k-1$) in a codeword $x\in \mathbb{Z}_{2k}^n$  by $n_i(x).$
 
The minimum Euclidean weight $d_E$ of $C$ is the smallest Euclidean weight among all nonzero
codewords of $C.$ 
A self-dual $\mathbb{Z}_{2k}$-code is called {\it Type II} if it has the property that every Euclidean
weight is divisible by $4k.$
Type II $\mathbb{Z}_{2k}$-codes are a remarkable
class of self-dual codes related to even unimodular lattices. 
There is a Type II $\mathbb{Z}_{2k}$-code of length $n$ if and only if $n$ is divisible by eight \cite{Bannai}. For
Type II $\mathbb{Z}_{2k}$-codes $C$ of length $n,$ the upper bound on the minimum Euclidean weight
\begin{equation}\label{de}
d_E(C) \leq 4k\left\lfloor\frac{n}{24} \right\rfloor+ 4k 
\end{equation}
holds for $k = 1$ and $2,$ and for $k\geq 3$ it holds under the assumption that $\left\lfloor\frac{n}{24} \right\rfloor\leq k-2$ (see \cite{Bannai}). We say that a Type II $\mathbb{Z}_{2k}$-code
meeting (\ref{de}) with equality is {\it extremal}. 

If $C$ is a $\mathbb{Z}_{2^m}$-code, then the code  $C^{(2^k)}=\{ x\ (\text{mod}\ 2^k)\, |\, x \in C\}$, $1 \leq k \leq m-1,$ is the {\it $\mathbb{Z}_{2^k}$-residue code} of  $C.$  
Each code $C$ over $\mathbb{Z}_{2^m}$ is permutation-equivalent to a code with a generator matrix in  {\it standard form}
$$\left(\begin{array}{ccccccc}
I_{k_1}&A_{1,2}&A_{1,3}&A_{1,4}&\cdots&\cdots&A_{1,m+1}\\
0&2I_{k_2}&2A_{2,3}&2A_{2,4}&\cdots&\cdots&2A_{2,m+1}\\
0&0&4I_{k_3}&4A_{3,4}&\cdots&\cdots&4A_{3,m+1}\\
\vdots&\vdots&\vdots&\ddots&\ddots& & \vdots\\
\vdots&\vdots&\vdots&\ddots&\ddots& & \vdots\\
0&0&0&\cdots&0&2^{m-1}I_{k_m}&2^{m-1}A_{m,m+1}
\end{array}\right),$$
where the matrix $A_{i,j}$ has elements in $\mathbb{Z}_{2^{j-1}}.$
We say that $C$ is of {\it type} $(k_1, k_2, k_3,\dots, k_m).$ The code $C$ has
$\prod_{j=1}^m (2^{m-j+1})^{k_j}$ codewords.\\

In this work, we have used computer algebra systems GAP \cite{cite_key1Gap} and Magma \cite{magma}.

\section{Method of construction} \label{constr}

In \cite{Chan}, the doubling method for a construction of Type II $\mathbb{Z}_4$-codes is introduced.
In the next theorem, we generalize results from \cite{Chan} and give the doubling method for a construction of Type II $\mathbb{Z}_{2k}$-codes. 

\begin{thm}\label{TmDoubling}
Let  $k\geq 2.$ Let $C$ be a Type II $\mathbb{Z}_{2k}$-code of length $n$ and let $n_i(x)$  denote the number of coordinates $i$  in $x\in \mathbb{Z}_{2k}^n$.
Let $ku\in \mathbb{Z}_{2k}^{n}\setminus C$ be a codeword with all coordinates equal to $0$ or $k$ with the following property:
if $k$ is odd, $n_k(ku)$ is divisible by four, if $k$ is even and not divisible by four, $n_k(ku)$ is even.
Let $C_0=\{v\in C\ |\ \left\langle ku,v\right\rangle=0\}.$
Then $\widetilde{C}=C_0 \oplus \left\langle ku\right\rangle$ is a Type II $\mathbb{Z}_{2k}$-code.
\end{thm}

\begin{proof}
The Euclidean weight $wt_E(ku)=n_k(ku)\cdot k^2$ is divisible by  $4k$ and $\left\langle ku,ku\right\rangle=0.$ It follows from (\ref{et}) that $\widetilde{C}$ is self-orthogonal with all Euclidean weights divisible by $4k.$

The codeword $ku\notin C,$ so there is a codeword $w\in C$ such that $\left\langle ku,w\right\rangle=k.$
Suppose $\tilde{w}\in C\setminus C_0.$ Then $\left\langle ku,\tilde{w}\right\rangle=k.$ Therefore, $\tilde{w}\in C_0+w$ and $C_0$ and $C_0+w$ are the only cosets of $C_0$ in $C.$ Since $|\widetilde{C}|=\left|C_0\right|\cdot 2=\left|C\right|,$ $\widetilde{C}$ is a Type II $\mathbb{Z}_{2k}$-code.
\end{proof}

When considering Type II $\mathbb{Z}_{2^m}$-codes, we can restrict the possible choices for $ku=2^{m-1}u\in \mathbb{Z}_{2^m}^{n}\setminus C$.

\begin{thm}\label{nule} 
Let $m\geq 2.$ Let $C$ be a Type II $\mathbb{Z}_{2^m}$-code of length $n$ and type $(k_1, k_2,\dots, k_m).$ The choice of $2^{m-1}u\in \mathbb{Z}_{2^m}^{n}\setminus C$ in Theorem \ref{TmDoubling} can be limited to codewords with zeroes on the first $k_1+k_2+\cdots +k_m$ coordinates.
\end{thm}
\begin{proof}
For every $2^{m-1}u\in \mathbb{Z}_{2^m}^{n}\setminus C$ satisfying the conditions of Theorem \ref{TmDoubling}, there exists a unique codeword $2^{m-1}v\in C$ with all coordinates equal to $0$ or $2^{m-1}$ such that $2^{m-1}u$ coincides
with $2^{m-1}v$ on the first $k_1+k_2+\cdots +k_m$ entries. Then $C_0 \oplus \left\langle 2^{m-1}u\right\rangle=C_0 \oplus \left\langle 2^{m-1}u-2^{m-1}v\right\rangle.$
\end{proof}

Now, we generalize the statement of   [\cite{40doubling}, Theorem 5] to Type II $\mathbb{Z}_{2^m}$-codes.

\begin{thm}\label{genMat}
Let $m\geq 2.$ Let $C$ be a Type II $\mathbb{Z}_{2^m}$-code of length $n$ and type $(k_1, k_2,\dots, k_m).$
Let $G$ be a generator matrix of $C$ in standard form and $G_i$ the $i^{th}$ row of $G.$
Let $2^{m-1}u\in \mathbb{Z}_{2^m}^{n}\setminus C$ be a codeword with zeroes on the first $k_1+k_2+\cdots +k_m$ coordinate positions such that $n_{2^{m-1}}(2^{m-1}u)$ is even if $m=2.$
Let $B=\{G_1,\dots,G_{k_1+k_2+\cdots +k_m}\}.$ The following process yields a generator matrix $\widetilde{G}$ of the $\mathbb{Z}_{2^m}$-code $\widetilde{C}$
obtained from $C$ and $2^{m-1}u$ by the doubling method.
\begin{enumerate}
\item[]{Step 1:} Let $B_E=\{G_i \in B\,|\,\left\langle G_i,2^{m-1}u\right\rangle=0\}$ and $B_O=B\setminus B_E.$
\item[]{Step 2:} Pick $G_i\in B_O$ arbitrarily. Define $B_O'=\{G_i+G_j\,|\,G_j\in B_O\}.$
\item[]{Step 3:} Let $\widetilde{G}$ be a matrix whose rows are the elements of the set $B_O'\cup B_E\cup \{2^{m-1}u\}.$
\end{enumerate}
The resultant code $\widetilde{C}$ is of type 
$$\left\{
\begin{array}{cc}
	(k_1-1, k_2+2),& \text{if}\ m=2,\\
	(k_1-1, k_2+1,k_3+1),& \text{if}\ m=3,\\
	(k_1-1, k_2+1,k_3,\dots, k_{m-1}, k_m+1),& \text{if}\ m\geq 4.
\end{array}\right.$$
The code $\widetilde{C}$ is independent of the choice of $G_i$ in Step 2.
\end{thm} 

\begin{proof}
The set $B_O$ is not empty because $C$ is self-dual and $2^{m-1}u\notin C.$
Further, for all $t=k_1+1,\dots,k_1+k_2+\cdots +k_m$ it follows $G_t\in B_E$. For $G_i, G_j\in B_O$, we have $$\left\langle G_i+G_j, 2^{m-1}u\right\rangle= \left\langle G_i, 2^{m-1}u\right\rangle +\left\langle G_j, 2^{m-1}u\right\rangle=0,$$
and $G_i+G_j$ is an codeword with all even coordinates if and only if  $i=j$.  

Note that $\left\langle B_O'\cup B_E \right\rangle =\{v\in C\,|\,\left\langle 2^{m-1}u,v\right\rangle=0\}$.
It follows that $\widetilde{C}$ is of type 
$$\left\{
\begin{array}{cc}
	(k_1-1, k_2+2),& \text{if}\ m=2,\\
	(k_1-1, k_2+1,k_3+1),& \text{if}\ m=3,\\
	(k_1-1, k_2+1,k_3,\dots, k_{m-1}, k_m+1),& \text{if}\ m\geq 4.
\end{array}\right.$$

The independence follows from the fact that $$G_k+G_j = (G_i +G_k)+(G_i +G_j)+(2^{m-1}-1)(G_i +G_i).$$
\end{proof}

\section{Construction of extremal Type II $\mathbb{Z}_8$-codes}
Here we consider an extremal Type II $\mathbb{Z}_8$-code $C$ of length $n \in \{24,32,40\}$ and type $(\frac{n}{2},0,0)$ which has an extremal residue code $C^{(4)}$. Using the doubling method given in the previous chapter, we developed an algorithm for a construction of extremal Type II $\mathbb{Z}_8$-codes $\widetilde{C}$ of length $n$ and type $(\frac{n}{2}-1,1,1).$

Note that $wt_E (x) = n_1(x)+n_7(x)+4(n_2(x)+n_6(x))+9(n_3(x)+n_5(x))+16n_4(x),$ for $x\in \mathbb{Z}_{8}^n.$

\begin{thm}\label{doublingZ8}
Let $n \in \{24,32,40\}$. Denote by $S_i(w)$ the set of positions with the element $i\in\mathbb{Z}_8$ in $w\in \mathbb{Z}_8^n.$ Let $C$ be an extremal Type II $\mathbb{Z}_8$-code of length $n$ and type $(k_1,k_2,k_3)$  
where $C^{(4)}$ is extremal.
Suppose $4u\in\mathbb{Z}_8^{n}$ is a codeword with all coordinates equal to $0$ or $4$ such that $S_4(4u)\subseteq \{k_1+k_2+k_3+1,\dots,n\},$ where $|S_4(4u)|\geq 2.$ 
If there is no codeword $v$ of $C$ that satisfies any of the following conditions:
\begin{enumerate}
\item $S_3(v)\cup S_4(v)\cup S_5(v)\subseteq S_4(4u)\subseteq S_2(v)\cup S_3(v)\cup S_4(v)\cup S_5(v)\cup S_6(v)$ and $$wt_E(v\ (\text{mod}\ 4))=16,$$
\item $|S_4(4u)\setminus S_4(v)|+|S_4(v)\setminus S_4(4u)|=1$ and $wt_E(v\ (\text{mod}\ 4))=0,$
\end{enumerate}
then the Type II $\mathbb{Z}_8$-code $\widetilde{C}$ generated by $4u$ and $C$ using the doubling method is extremal. These choices of $4u$ are the only candidates for the code $C$ in the doubling method which lead to an extremal code.
\end{thm}

\begin{proof}  
It follows from Theorem \ref{nule} that the choices of $4u$ can be limited to codewords with zeroes on  the first $k_1+k_2+k_3$ coordinates.
Let us assume that the code $\widetilde{C}$ is not extremal. Then it contains a codeword of Euclidean weight $16$ of the form  $w=v+4u,$ where $v\in C$ is such a codeword that $\left\langle v,4u\right\rangle=0,$ and 
$$wt_E(w)=W+8(n_3(w)+n_5(w))+16n_4(w),$$
where $W=n_1(w)+n_7(w)+n_3(w)+n_5(w)+4(n_2(w)+n_6(w)).$
It holds $W=0$ or $W\geq 16,$ since $C^{(4)}$ is extremal.
There are three cases to consider.\\
Case 1: $W>16.$\\
Then $wt_E(w)\geq 32.$\\
Case 2: $W=16.$\\
For $wt_E(w)$ to be equal to $16,$ $n_3(w)=n_4(w)=n_5(w)=0,$ which is impossible because of the first condition.\\
Case 3: $W=0.$\\
This condition implies that $w$ and $v$ have all coordinates equal to $0$ or $4.$
Then, for $wt_E(w)=16n_4(w)$ to be equal to $16,$   $n_4(w)=1,$ which is impossible because of the second condition.

The resulting choices for $4u$ are the only candidates for the code $C$ in the doubling method, since the conditions of the theorem exclude all choices that lead to a code $\widetilde{C}$ which is not extremal.
\end{proof}

For an extremal Type II $\mathbb{Z}_8$-code $C$  of length $n \in \{24,32,40\}$ and type $(\frac{n}{2},0,0)$ which has an extremal residue code $C^{(4)},$ the next algorithm returns all unsuitable candidates $4u$,  i.e., the candidates for which  the application of the doubling method leads to a Type II $\mathbb{Z}_8$-code $\widetilde{C}$ which is not extremal. Thus, performing the given steps will find all possible candidates $4u$ for code $C$ to produce a new extremal Type II $\mathbb{Z}_8$-code $\widetilde{C}$ by the doubling method.\\

 {\bf Algorithm C} \\\label{doubling8alg}
Let $n \in \{24,32,40\}.$ Denote by $S_i(w)$ the set of positions with the element $i\in\mathbb{Z}_8$ in $w\in \mathbb{Z}_8^n.$  Let $C$ be an extremal Type II $\mathbb{Z}_8$-code of length $n$ and type $(\frac{n}{2},0,0),$ where $C^{(4)}$ is extremal, with the generator matrix 
$G=\left[\begin{tabular}{cc}
$I_{\frac{n}{2}}$&$A$
\end{tabular}\right]$ in the standard form.

\begin{itemize}
\item[1.] Let $v=(v_1,\dots,v_n)\in C^{(4)}$ be a codeword of Euclidean weight $16.$
\item[1.2.] Let $F_v=\{v_1,\dots,v_{\frac{n}{2}}\},\ A_v=S_2(v)\cap F_v$ and $B_v=S_3(v)\cap F_v.$ 
\item[1.3.] Repeat the following steps on all $A\subseteq A_v:$
\begin{itemize}
\item[1.3.1.] Calculate $v'=v+4s_A+4s_{B_v},$ where $s_A$ is the sum of rows in the generator matrix $G$ of $C$ with row indices in $A$ and $s_{B_v}$ is the sum of rows in $G$ with row indices in $B_v.$
\item[1.3.2.] Let $$O_{v'}=(S_2(v')\cup S_6(v'))\cap \left\{\frac{n}{2}+1,\dots,n\right\},$$ 
$$P_{v'}=S_4(v')\cap \left\{\frac{n}{2}+1,\dots,n\right\},$$
$$Q_{v'}=(S_3(v')\cup S_5(v'))\cap \left\{\frac{n}{2}+1,\dots,n\right\}.$$
\item[1.3.3.] Let $\mathcal{B}$ be the collection of all sets  $$B=O\cup P_{v'}\cup Q_{v'},\ O\subseteq O_{v'},$$ where $|B|\geq 2.$
 \end{itemize}
\item[2.] For all $i\in \{1,\dots,\frac{n}{2}\},$ do the following.
\begin{itemize}
\item[2.1.] Let $O_i=S_4(4G_i)\cap \left\{\frac{n}{2}+1,\dots,n\right\},$ where $G_i$ is the $i^{th}$ row of $G.$
\item[2.2.] Include all  $O_i$ such that $|O_i|\geq 2$ in $\mathcal{B}.$
\end{itemize}
\end{itemize}

\noindent Our method of construction is based on the following theorem.

\begin{thm}\label{TmAlgC}
Let $n \in \{24,32,40\}.$ Denote by $S_i(w)$ the set of positions with the element $i\in\mathbb{Z}_8$ in $w\in \mathbb{Z}_8^n.$ Let $C$ be an extremal Type II $\mathbb{Z}_8$-code of length $n$ and type $(\frac{n}{2},0,0),$ where $C^{(4)}$ is extremal.
Furthermore, let $\mathcal{S}$ be the collection of all $S\subseteq\left\{\frac{n}{2}+1,\dots,n\right\}$ such that $|S|\geq 2.$
Then $\mathcal{G}=\mathcal{S}\setminus \mathcal{B}$ is the set of all possible $S_4(4u)$ for the code $C$ in the doubling method which lead to an extremal Type II $\mathbb{Z}_8$-code $\widetilde{C}$ of length $n$ and type $(\frac{n}{2}-1,1,1),$ where $\mathcal{B}$ is the set obtained by applying Algorithm C.
\end{thm}

\begin{proof} The first condition in Theorem \ref{doublingZ8} is checked in Step 1. of Algorithm C. Since the condition requires that $S_3(v)\cup S_4(v)\cup S_5(v)\subseteq S_4(4u),$ the coefficients of the rows of  $G$ in the linear combination of $v$ cannot be $3, 4$ or $5.$ Step 1.3.1. generates all such codewords $v'$ with $wt_E(v'\ (\text{mod}\ 4))=16.$ All subsets $B=S_4(4u)$ satisfying the first condition are included in $\mathcal{B}$ in Step 1.3.3.

The second condition of Theorem \ref{doublingZ8}, implies that $v$ is one of the rows of $G$ with coefficient $4.$
These codewords are considered in Step 2.
 
\end{proof}

\subsection{New extremal Type II $\mathbb{Z}_8$-codes of length $32$} \label{results}

Six inequivalent extremal Type II $\mathbb{Z}_8$-codes of length $32$ are known: $C_{8,32,i},\ i=1,\dots,5$ from \cite{Hnovi} and  $D_{32}$ from \cite{6810}.   Since all of them are of type  $(16,0,0)$, we investigate the possibility of constructing new extremal Type II $\mathbb{Z}_8$-codes of length $32$  by using the introduced doubling method. The result of our analysis is given in Proposition \ref{newZ8}.

\begin{prop} \label{newZ8}
There are at least $10$ inequivalent extremal Type II $\mathbb{Z}_8$-codes of length $32$ and type $(15,1,1).$
\end{prop}

\begin{proof}
The extremal Type II $\mathbb{Z}_8$-codes $C_{8,32,1}$ and  $C_{8,32,2}$ are of type $(16,0,0)$ and length $32$ and have extremal $\mathbb{Z}_4$-residue codes. So, we can apply Theorem \ref{TmAlgC}. We applied Algorithm C and found $23067$ candidates $4u$ for a construction of extremal Type II $\mathbb{Z}_8$-codes of type $(15,1,1)$ and length $32$ from $C_{8,32,1}$ by doubling method. Also, we found $22818$ candidates $4u$ for a construction of extremal Type II $\mathbb{Z}_8$-codes of type $(15,1,1)$ and length $32$ from $C_{8,32,2}$ by doubling method. 

The extremal Type II $\mathbb{Z}_8$-codes $C_{8,32,i},\ i=3,4,5$ are of type $(16,0,0)$ and length $32$ and have $\mathbb{Z}_4$-residue codes of minimum Euclidean weight $8.$ So, we cannot apply Theorem \ref{TmAlgC} to obtain new extremal Type II $\mathbb{Z}_8$-codes from $C_{8,32,i},\ i=3,4,5$ using the doubling method.

Further, the extremal Type II $\mathbb{Z}_8$-code $D_{32}$ is of type $(16,0,0)$ and length $32$ and it has an extremal $\mathbb{Z}_4$-residue code. So, we can apply Theorem \ref{TmAlgC}. We applied Algorithm C and found $22965$ candidates $4u$ for a construction of extremal Type II $\mathbb{Z}_8$-codes of type $(15,1,1)$ and length $32$ from $D_{32}$ by doubling method.

We use Theorem \ref{genMat} to obtain the generator matrices for  
the $68850$ constructed extremal Type II $\mathbb{Z}_8$-codes of type $(15,1,1)$ and length $32.$ 
Using Magma (\cite{magma}), we calculated the weight distributions of the corresponding $68850$ binary residue codes and obtained that, with respect to the weight distribution of their binary residue codes, all constructed extremal Type II $\mathbb{Z}_8$-codes are distributed into $10$ classes.  
The corresponding weight distributions are given in Table \ref{tablica}.

\begin{table}[h] 
	\centering
		\begin{tabular}{|c|c|c|c|c|c|c|c|c|}
		\hline
			&$i$& 0&8&12&16&20&24&32\\
			\hline
			${C_1}^{(2)}$&$W_i$& 1& 316&6912 &18310 &6912 &316 & 1\\
			\hline
			${C_2}^{(2)}$&$W_i$& 1&332 &6848 & 18406&6848 &332 & 1\\
			\hline
			${C_3}^{(2)}$&$W_i$&1&337 &6888 &18259 &7000 & 283&0 \\
			\hline
			${C_4}^{(2)}$&$W_i$&1&305 &6952 &18259 &6936 &315 &0\\
			\hline
			${C_5}^{(2)}$&$W_i$&1&308 &6944 &18262 &6944 &308 &1 \\
			\hline
			${C_6}^{(2)}$&$W_i$&1&300 &6976 &18214 &6976 &300 &1 \\
			\hline
			${C_7}^{(2)}$&$W_i$&1&364 &6720 &18598 &6720 &364 &1 \\
			\hline
			${C_8}^{(2)}$&$W_i$&1&380 &7168 &17670 &7168 &380 &1 \\
			\hline
			${C_9}^{(2)}$&$W_i$&1&324 &6880 &18358 &6880 &324 &1 \\
			\hline
			${C_{10}}^{(2)}$&$W_i$&1&340 &6816 &18454 &6816 &340 &1\\
			\hline
		\end{tabular}\caption{Weight distributions of the binary residue codes}\label{tablica} 
\end{table}

For each of the obtained weight distribution classes we give the generator matrix in standard form of one extremal Type II $\mathbb{Z}_8$-code of length $32$ and type $(15,1,1)$, namely, the generator matrices for the following class representatives:
$$C_1= {C_{8,32,1}}_0 \oplus \left\langle 4u\right\rangle,\ S_4(4u)=\{ 17, 19, 21, 22 \},$$
$$C_2= {C_{8,32,1}}_0 \oplus \left\langle 4u\right\rangle,\ S_4(4u)=\{ 17, 18, 20, 21 \},$$
$$C_3= {C_{8,32,1}}_0 \oplus \left\langle 4u\right\rangle,\ S_4(4u)=\{ 17, 19, 21 \},$$
$$C_4= {C_{8,32,1}}_0 \oplus \left\langle 4u\right\rangle,\ S_4(4u)=\{ 17, 19, 20, 21, 22 \},$$  
$$C_5= {C_{8,32,1}}_0 \oplus \left\langle 4u\right\rangle,\ S_4(4u)=\{ 17, 18, 19, 20 \},$$  
$$C_6= {C_{8,32,1}}_0 \oplus \left\langle 4u\right\rangle,\ S_4(4u)=\{ 17, 18, 19, 20, 21, 22 \},$$
$$C_7= {C_{8,32,1}}_0 \oplus \left\langle 4u\right\rangle,\ S_4(4u)=\{ 18, 24, 25, 27 \},$$
$$C_8= {C_{8,32,1}}_0 \oplus \left\langle 4u\right\rangle,\ S_4(4u)=\{ 20, 25 \},$$   
$$C_9= {C_{8,32,2}}_0 \oplus \left\langle 4u\right\rangle,\ S_4(4u)=\{ 17, 18, 19, 21 \},$$
$$C_{10}= {C_{8,32,2}}_0 \oplus \left\langle 4u\right\rangle,\ S_4(4u)=\{ 17, 21, 24, 25 \}.$$  

Those are, respectively:
{\scriptsize $$G_1=\left(\begin{array}{c}
	1 0 0 0 0 0 0 0 0 0 0 0 0 0 0 0 3 4 7 6 7 1 6 3 5 6 0 2 0 4 7 4\\
    0 1 0 0 0 0 0 0 0 0 0 0 0 0 0 0 1 7 0 3 2 7 5 6 4 5 6 0 2 0 4 7\\
    0 0 1 0 0 0 0 0 0 0 0 0 0 0 0 0 2 5 7 4 3 6 3 5 1 4 5 6 0 2 0 4\\
    0 0 0 1 0 0 0 0 0 0 0 0 0 0 0 1 0 6 1 5 0 1 1 7 2 6 7 3 1 4 1 5\\
    0 0 0 0 1 0 0 0 0 0 0 0 0 0 0 0 1 3 6 5 3 0 3 6 0 4 1 4 5 6 0 2\\
    0 0 0 0 0 1 0 0 0 0 0 0 0 0 0 1 3 1 3 0 5 1 7 7 4 5 7 7 7 1 5 5\\
    0 0 0 0 0 0 1 0 0 0 0 0 0 0 0 1 2 6 1 1 6 3 2 0 6 3 3 2 4 0 4 3\\
    0 0 0 0 0 0 0 1 0 0 0 0 0 0 0 1 1 1 6 3 7 0 0 3 0 5 1 6 7 5 3 2\\
    0 0 0 0 0 0 0 0 1 0 0 0 0 0 0 0 3 4 7 4 0 2 0 6 7 5 2 3 5 2 5 4\\
    0 0 0 0 0 0 0 0 0 1 0 0 0 0 0 0 2 3 4 7 4 0 2 0 4 7 5 2 3 5 2 5\\
    0 0 0 0 0 0 0 0 0 0 1 0 0 0 0 0 0 2 3 4 7 4 0 2 3 4 7 5 2 3 5 2\\
    0 0 0 0 0 0 0 0 0 0 0 1 0 0 0 1 3 4 6 1 0 5 7 4 4 0 7 5 0 6 2 2\\
    0 0 0 0 0 0 0 0 0 0 0 0 1 0 0 1 1 2 0 0 3 6 2 0 1 3 6 2 2 1 1 0\\
    0 0 0 0 0 0 0 0 0 0 0 0 0 1 0 1 1 4 2 6 6 1 7 3 3 0 1 1 7 3 4 7\\
    0 0 0 0 0 0 0 0 0 0 0 0 0 0 1 0 0 1 0 0 2 4 2 3 3 6 5 3 6 3 4 7\\
    0 0 0 0 0 0 0 0 0 0 0 0 0 0 0 2 2 0 0 4 0 4 6 0 4 2 6 4 6 0 6 2\\
    0 0 0 0 0 0 0 0 0 0 0 0 0 0 0 0 4 0 4 0 4 4 0 0 0 0 0 0 0 0 0 0
\end{array}\right),\
G_2=\left(\begin{array}{c}
	1 0 0 0 0 0 0 0 0 0 0 0 0 0 0 0 3 0 3 2 7 5 6 3 5 6 0 2 0 4 7 4\\
    0 1 0 0 0 0 0 0 0 0 0 0 0 0 0 1 1 0 0 3 4 3 7 1 7 3 3 3 0 3 0 6\\
    0 0 1 0 0 0 0 0 0 0 0 0 0 0 0 0 2 5 7 4 3 6 3 5 1 4 5 6 0 2 0 4\\
    0 0 0 1 0 0 0 0 0 0 0 0 0 0 0 0 3 2 5 7 4 3 6 3 4 1 4 5 6 0 2 0\\
    0 0 0 0 1 0 0 0 0 0 0 0 0 0 0 0 1 7 2 1 3 4 3 6 0 4 1 4 5 6 0 2\\
    0 0 0 0 0 1 0 0 0 0 0 0 0 0 0 0 2 5 3 2 5 7 4 3 6 0 4 1 4 5 6 0\\
    0 0 0 0 0 0 1 0 0 0 0 0 0 0 0 0 1 6 5 7 6 5 7 4 0 6 0 4 1 4 5 6\\
    0 0 0 0 0 0 0 1 0 0 0 0 0 0 0 1 0 6 6 5 1 2 7 2 5 6 3 3 2 4 0 4\\
    0 0 0 0 0 0 0 0 1 0 0 0 0 0 0 1 3 1 3 0 2 2 2 1 2 3 7 6 3 5 1 3\\
    0 0 0 0 0 0 0 0 0 1 0 0 0 0 0 0 2 3 4 7 4 0 2 0 4 7 5 2 3 5 2 5\\
    0 0 0 0 0 0 0 0 0 0 1 0 0 0 0 1 0 7 7 0 1 4 2 5 6 2 4 0 0 6 1 1\\
    0 0 0 0 0 0 0 0 0 0 0 1 0 0 0 1 2 1 6 3 2 7 6 3 1 1 1 2 3 5 7 4\\
    0 0 0 0 0 0 0 0 0 0 0 0 1 0 0 1 0 3 4 6 5 4 1 7 6 4 0 7 5 0 6 2\\
    0 0 0 0 0 0 0 0 0 0 0 0 0 1 0 0 0 4 6 4 6 3 4 7 5 3 6 3 4 7 5 2\\
    0 0 0 0 0 0 0 0 0 0 0 0 0 0 1 1 1 1 4 2 2 2 5 7 1 3 0 1 1 7 3 4\\
    0 0 0 0 0 0 0 0 0 0 0 0 0 0 0 2 0 2 0 0 4 0 4 6 6 4 2 6 4 6 0 6\\
    0 0 0 0 0 0 0 0 0 0 0 0 0 0 0 0 4 4 0 4 4 0 0 0 0 0 0 0 0 0 0 0\\
\end{array}\right),$$
$$G_3=\left(\begin{array}{c}
	1 0 0 0 0 0 0 0 0 0 0 0 0 0 0 1 0 0 3 4 3 7 1 7 3 3 3 0 3 0 6 1\\
    0 1 0 0 0 0 0 0 0 0 0 0 0 0 0 1 2 3 0 1 2 5 0 2 2 2 1 6 5 4 3 4\\
    0 0 1 0 0 0 0 0 0 0 0 0 0 0 0 0 2 5 7 4 3 6 3 5 1 4 5 6 0 2 0 4\\
    0 0 0 1 0 0 0 0 0 0 0 0 0 0 0 0 3 2 5 7 4 3 6 3 4 1 4 5 6 0 2 0\\
    0 0 0 0 1 0 0 0 0 0 0 0 0 0 0 0 1 3 6 5 3 4 3 6 0 4 1 4 5 6 0 2\\
    0 0 0 0 0 1 0 0 0 0 0 0 0 0 0 0 2 5 3 2 5 7 4 3 6 0 4 1 4 5 6 0\\
    0 0 0 0 0 0 1 0 0 0 0 0 0 0 0 0 1 2 1 3 6 5 7 4 0 6 0 4 1 4 5 6\\
    0 0 0 0 0 0 0 1 0 0 0 0 0 0 0 1 1 1 6 3 7 4 0 3 0 5 1 6 7 5 3 2\\
    0 0 0 0 0 0 0 0 1 0 0 0 0 0 0 0 3 4 7 4 0 2 0 6 7 5 2 3 5 2 5 4\\
    0 0 0 0 0 0 0 0 0 1 0 0 0 0 0 0 2 3 4 7 4 0 2 0 4 7 5 2 3 5 2 5\\
    0 0 0 0 0 0 0 0 0 0 1 0 0 0 0 0 0 2 3 4 7 4 0 2 3 4 7 5 2 3 5 2\\
    0 0 0 0 0 0 0 0 0 0 0 1 0 0 0 0 2 0 6 3 0 7 4 0 6 3 4 7 5 2 3 5\\
    0 0 0 0 0 0 0 0 0 0 0 0 1 0 0 1 1 2 0 0 3 6 2 0 1 3 6 2 2 1 1 0\\
    0 0 0 0 0 0 0 0 0 0 0 0 0 1 0 0 0 0 2 0 6 3 4 7 5 3 6 3 4 7 5 2\\
    0 0 0 0 0 0 0 0 0 0 0 0 0 0 1 0 0 1 0 0 2 0 2 3 3 6 5 3 6 3 4 7\\
    0 0 0 0 0 0 0 0 0 0 0 0 0 0 0 2 2 0 0 4 0 4 6 0 4 2 6 4 6 0 6 2\\
    0 0 0 0 0 0 0 0 0 0 0 0 0 0 0 0 4 0 4 0 4 0 0 0 0 0 0 0 0 0 0 0\\
\end{array}\right),\
G_4=\left(\begin{array}{c}
	1 0 0 0 0 0 0 0 0 0 0 0 0 0 0 0 3 4 7 2 7 1 6 3 5 6 0 2 0 4 7 4\\
    0 1 0 0 0 0 0 0 0 0 0 0 0 0 0 1 2 3 0 5 2 1 0 2 2 2 1 6 5 4 3 4\\
    0 0 1 0 0 0 0 0 0 0 0 0 0 0 0 0 2 5 7 4 3 6 3 5 1 4 5 6 0 2 0 4\\
    0 0 0 1 0 0 0 0 0 0 0 0 0 0 0 0 3 2 5 7 4 3 6 3 4 1 4 5 6 0 2 0\\
    0 0 0 0 1 0 0 0 0 0 0 0 0 0 0 1 2 7 6 7 3 2 6 2 6 1 4 2 0 2 7 7\\
    0 0 0 0 0 1 0 0 0 0 0 0 0 0 0 1 3 1 3 0 5 1 7 7 4 5 7 7 7 1 5 5\\
    0 0 0 0 0 0 1 0 0 0 0 0 0 0 0 0 1 2 1 7 6 1 7 4 0 6 0 4 1 4 5 6\\
    0 0 0 0 0 0 0 1 0 0 0 0 0 0 0 0 0 5 6 1 7 6 5 7 2 0 6 0 4 1 4 5\\
    0 0 0 0 0 0 0 0 1 0 0 0 0 0 0 0 3 4 7 4 0 2 0 6 7 5 2 3 5 2 5 4\\
    0 0 0 0 0 0 0 0 0 1 0 0 0 0 0 1 3 7 4 5 4 2 5 4 2 4 0 0 6 1 1 2\\
    0 0 0 0 0 0 0 0 0 0 1 0 0 0 0 0 0 2 3 4 7 4 0 2 3 4 7 5 2 3 5 2\\
    0 0 0 0 0 0 0 0 0 0 0 1 0 0 0 0 2 0 6 7 0 3 4 0 6 3 4 7 5 2 3 5\\
    0 0 0 0 0 0 0 0 0 0 0 0 1 0 0 1 1 2 0 0 3 6 2 0 1 3 6 2 2 1 1 0\\
    0 0 0 0 0 0 0 0 0 0 0 0 0 1 0 1 1 4 2 2 6 1 7 3 3 0 1 1 7 3 4 7\\
    0 0 0 0 0 0 0 0 0 0 0 0 0 0 1 0 0 1 0 4 2 4 2 3 3 6 5 3 6 3 4 7\\
    0 0 0 0 0 0 0 0 0 0 0 0 0 0 0 2 2 0 0 4 0 4 6 0 4 2 6 4 6 0 6 2\\
    0 0 0 0 0 0 0 0 0 0 0 0 0 0 0 0 4 0 4 4 4 4 0 0 0 0 0 0 0 0 0 0\\
\end{array}\right),$$
$$G_5=\left(\begin{array}{c}
	1 0 0 0 0 0 0 0 0 0 0 0 0 0 0 0 3 0 7 2 3 5 6 3 5 6 0 2 0 4 7 4\\
    0 1 0 0 0 0 0 0 0 0 0 0 0 0 0 1 1 0 0 3 4 3 7 1 7 3 3 3 0 3 0 6\\
    0 0 1 0 0 0 0 0 0 0 0 0 0 0 0 0 2 5 7 4 3 6 3 5 1 4 5 6 0 2 0 4\\
    0 0 0 1 0 0 0 0 0 0 0 0 0 0 0 1 3 7 5 3 2 3 0 6 7 7 1 0 4 3 6 7\\
    0 0 0 0 1 0 0 0 0 0 0 0 0 0 0 1 1 4 6 5 5 4 5 1 3 2 6 7 3 1 4 1\\
    0 0 0 0 0 1 0 0 0 0 0 0 0 0 0 0 2 5 3 2 5 7 4 3 6 0 4 1 4 5 6 0\\
    0 0 0 0 0 0 1 0 0 0 0 0 0 0 0 1 1 3 1 3 0 5 1 7 3 4 5 7 7 7 1 5\\
    0 0 0 0 0 0 0 1 0 0 0 0 0 0 0 0 0 1 6 1 3 2 5 7 2 0 6 0 4 1 4 5\\
    0 0 0 0 0 0 0 0 1 0 0 0 0 0 0 0 3 4 7 4 0 2 0 6 7 5 2 3 5 2 5 4\\
    0 0 0 0 0 0 0 0 0 1 0 0 0 0 0 0 2 3 4 7 4 0 2 0 4 7 5 2 3 5 2 5\\
    0 0 0 0 0 0 0 0 0 0 1 0 0 0 0 1 0 7 3 0 5 4 2 5 6 2 4 0 0 6 1 1\\
    0 0 0 0 0 0 0 0 0 0 0 1 0 0 0 1 2 1 6 3 2 7 6 3 1 1 1 2 3 5 7 4\\
    0 0 0 0 0 0 0 0 0 0 0 0 1 0 0 0 0 6 0 2 3 4 7 4 3 6 3 4 7 5 2 3\\
    0 0 0 0 0 0 0 0 0 0 0 0 0 1 0 0 0 4 2 4 2 3 4 7 5 3 6 3 4 7 5 2\\
    0 0 0 0 0 0 0 0 0 0 0 0 0 0 1 1 1 1 0 2 6 2 5 7 1 3 0 1 1 7 3 4\\
    0 0 0 0 0 0 0 0 0 0 0 0 0 0 0 2 0 2 0 0 4 0 4 6 6 4 2 6 4 6 0 6\\
    0 0 0 0 0 0 0 0 0 0 0 0 0 0 0 0 4 4 4 4 0 0 0 0 0 0 0 0 0 0 0 0\\
\end{array}\right),\
G_6=\left(\begin{array}{c}
	 1 0 0 0 0 0 0 0 0 0 0 0 0 0 0 0 3 0 7 2 7 1 6 3 5 6 0 2 0 4 7 4\\
    0 1 0 0 0 0 0 0 0 0 0 0 0 0 0 0 1 3 0 7 2 7 5 6 4 5 6 0 2 0 4 7\\
    0 0 1 0 0 0 0 0 0 0 0 0 0 0 0 1 2 2 7 0 5 2 5 0 4 2 2 1 6 5 4 3\\
    0 0 0 1 0 0 0 0 0 0 0 0 0 0 0 0 3 2 5 7 4 3 6 3 4 1 4 5 6 0 2 0\\
    0 0 0 0 1 0 0 0 0 0 0 0 0 0 0 0 1 7 6 1 3 0 3 6 0 4 1 4 5 6 0 2\\
    0 0 0 0 0 1 0 0 0 0 0 0 0 0 0 0 2 5 3 2 5 7 4 3 6 0 4 1 4 5 6 0\\
    0 0 0 0 0 0 1 0 0 0 0 0 0 0 0 0 1 6 1 7 6 1 7 4 0 6 0 4 1 4 5 6\\
    0 0 0 0 0 0 0 1 0 0 0 0 0 0 0 1 0 6 6 5 1 2 7 2 5 6 3 3 2 4 0 4\\
    0 0 0 0 0 0 0 0 1 0 0 0 0 0 0 0 3 4 7 4 0 2 0 6 7 5 2 3 5 2 5 4\\
    0 0 0 0 0 0 0 0 0 1 0 0 0 0 0 0 2 3 4 7 4 0 2 0 4 7 5 2 3 5 2 5\\
    0 0 0 0 0 0 0 0 0 0 1 0 0 0 0 0 0 2 3 4 7 4 0 2 3 4 7 5 2 3 5 2\\
    0 0 0 0 0 0 0 0 0 0 0 1 0 0 0 0 2 4 6 7 0 3 4 0 6 3 4 7 5 2 3 5\\
    0 0 0 0 0 0 0 0 0 0 0 0 1 0 0 1 0 3 0 6 5 0 1 7 6 4 0 7 5 0 6 2\\
    0 0 0 0 0 0 0 0 0 0 0 0 0 1 0 1 0 1 2 0 0 3 6 2 0 1 3 6 2 2 1 1\\
    0 0 0 0 0 0 0 0 0 0 0 0 0 0 1 1 1 1 0 2 2 6 5 7 1 3 0 1 1 7 3 4\\
    0 0 0 0 0 0 0 0 0 0 0 0 0 0 0 2 0 2 0 0 4 0 4 6 6 4 2 6 4 6 0 6\\
    0 0 0 0 0 0 0 0 0 0 0 0 0 0 0 0 4 4 4 4 4 4 0 0 0 0 0 0 0 0 0 0\\
\end{array}\right),$$
$$G_7=\left(\begin{array}{c}
	 1 0 0 0 0 0 0 0 0 0 0 0 0 0 0 0 0 7 3 6 3 5 6 4 3 6 5 2 0 4 7 4\\
    0 1 0 0 0 0 0 0 0 0 0 0 0 0 0 1 1 6 4 1 6 5 0 3 2 2 2 6 5 4 3 4\\
    0 0 1 0 0 0 0 0 0 0 0 0 0 0 0 0 1 2 7 4 3 6 3 1 1 4 5 6 0 2 0 4\\
    0 0 0 1 0 0 0 0 0 0 0 0 0 0 0 1 3 4 5 5 4 5 1 2 3 6 6 3 1 4 1 5\\
    0 0 0 0 1 0 0 0 0 0 0 0 0 0 0 0 1 5 2 5 7 4 3 3 6 4 0 4 5 6 0 2\\
    0 0 0 0 0 1 0 0 0 0 0 0 0 0 0 0 0 2 3 2 5 7 4 1 7 0 2 1 4 5 6 0\\
    0 0 0 0 0 0 1 0 0 0 0 0 0 0 0 0 0 5 5 3 2 5 7 2 4 6 0 4 1 4 5 6\\
    0 0 0 0 0 0 0 1 0 0 0 0 0 0 0 0 2 4 2 5 3 2 5 1 3 0 6 0 4 1 4 5\\
    0 0 0 0 0 0 0 0 1 0 0 0 0 0 0 1 1 4 7 2 0 4 3 4 6 2 1 1 0 6 4 1\\
    0 0 0 0 0 0 0 0 0 1 0 0 0 0 0 0 1 2 4 7 4 0 2 7 4 7 0 2 3 5 2 5\\
    0 0 0 0 0 0 0 0 0 0 1 0 0 0 0 0 3 0 3 4 7 4 0 6 6 4 7 5 2 3 5 2\\
    0 0 0 0 0 0 0 0 0 0 0 1 0 0 0 0 0 6 2 3 4 7 4 4 4 3 2 7 5 2 3 5\\
    0 0 0 0 0 0 0 0 0 0 0 0 1 0 0 0 3 0 0 2 3 4 7 6 4 6 3 4 7 5 2 3\\
    0 0 0 0 0 0 0 0 0 0 0 0 0 1 0 0 2 4 6 0 2 3 4 4 3 3 1 3 4 7 5 2\\
    0 0 0 0 0 0 0 0 0 0 0 0 0 0 1 0 1 4 4 0 6 0 2 5 7 6 7 3 6 3 4 7\\
    0 0 0 0 0 0 0 0 0 0 0 0 0 0 0 2 2 2 0 4 0 4 6 4 4 2 0 4 6 0 6 2\\
    0 0 0 0 0 0 0 0 0 0 0 0 0 0 0 0 4 0 0 0 0 0 0 4 4 0 4 0 0 0 0 0\\
\end{array}\right),\
G_8=\left(\begin{array}{c}
	1 0 0 0 0 0 0 0 0 0 0 0 0 0 0 1 0 5 7 6 1 5 0 6 3 4 5 5 6 7 3 3\\
    0 1 0 0 0 0 0 0 0 0 0 0 0 0 0 1 3 0 0 7 4 3 7 1 1 3 3 3 0 3 0 6\\
    0 0 1 0 0 0 0 0 0 0 0 0 0 0 0 1 0 6 3 0 1 6 5 0 6 2 2 1 6 5 4 3\\
    0 0 0 1 0 0 0 0 0 0 0 0 0 0 0 1 3 3 1 3 2 3 0 6 7 7 1 0 4 3 6 7\\
    0 0 0 0 1 0 0 0 0 0 0 0 0 0 0 1 3 4 6 5 5 4 5 1 1 2 6 7 3 1 4 1\\
    0 0 0 0 0 1 0 0 0 0 0 0 0 0 0 0 2 5 3 6 5 7 4 3 2 0 4 1 4 5 6 0\\
    0 0 0 0 0 0 1 0 0 0 0 0 0 0 0 1 3 3 1 3 0 5 1 7 1 4 5 7 7 7 1 5\\
    0 0 0 0 0 0 0 1 0 0 0 0 0 0 0 1 1 6 6 1 1 2 7 2 0 6 3 3 2 4 0 4\\
    0 0 0 0 0 0 0 0 1 0 0 0 0 0 0 1 2 5 3 4 6 2 2 1 7 3 7 6 3 5 1 3\\
    0 0 0 0 0 0 0 0 0 1 0 0 0 0 0 1 3 4 0 3 2 0 4 3 6 5 2 5 1 0 6 4\\
    0 0 0 0 0 0 0 0 0 0 1 0 0 0 0 1 2 3 7 0 5 4 2 5 4 2 4 0 0 6 1 1\\
    0 0 0 0 0 0 0 0 0 0 0 1 0 0 0 1 1 1 6 3 2 7 6 3 2 1 1 2 3 5 7 4\\
    0 0 0 0 0 0 0 0 0 0 0 0 1 0 0 1 2 7 4 6 1 4 1 7 4 4 0 7 5 0 6 2\\
    0 0 0 0 0 0 0 0 0 0 0 0 0 1 0 1 0 1 2 0 0 3 6 2 0 1 3 6 2 2 1 1\\
    0 0 0 0 0 0 0 0 0 0 0 0 0 0 1 0 2 4 0 2 0 2 3 4 1 5 3 6 3 4 7 5\\
    0 0 0 0 0 0 0 0 0 0 0 0 0 0 0 2 2 2 0 4 4 0 4 6 0 4 2 6 4 6 0 6\\
    0 0 0 0 0 0 0 0 0 0 0 0 0 0 0 0 4 0 0 4 0 0 0 0 0 0 0 0 0 0 0 0\\
\end{array}\right),$$
$$G_9=\left(\begin{array}{c}
	 1 0 0 0 0 0 0 0 0 0 0 0 0 0 0 0 3 0 7 6 0 5 2 5 6 5 2 2 0 7 1 2\\
    0 1 0 0 0 0 0 0 0 0 0 0 0 0 0 0 3 7 4 3 6 4 5 2 6 6 5 2 2 0 7 1\\
    0 0 1 0 0 0 0 0 0 0 0 0 0 0 0 1 0 5 6 6 1 7 3 5 5 1 2 4 5 0 3 3\\
    0 0 0 1 0 0 0 0 0 0 0 0 0 0 0 0 3 6 3 7 4 3 6 4 1 7 6 6 5 2 2 0\\
    0 0 0 0 1 0 0 0 0 0 0 0 0 0 0 0 0 7 2 3 3 4 3 6 0 1 7 6 6 5 2 2\\
    0 0 0 0 0 1 0 0 0 0 0 0 0 0 0 0 2 4 3 6 3 7 4 3 6 0 1 7 6 6 5 2\\
    0 0 0 0 0 0 1 0 0 0 0 0 0 0 0 1 3 0 7 5 0 4 6 4 4 1 4 0 2 4 1 1\\
    0 0 0 0 0 0 0 1 0 0 0 0 0 0 0 0 0 1 6 4 7 6 3 7 3 6 6 0 1 7 6 6\\
    0 0 0 0 0 0 0 0 1 0 0 0 0 0 0 1 0 0 4 1 2 3 1 5 5 6 2 2 7 0 0 0\\
    0 0 0 0 0 0 0 0 0 1 0 0 0 0 0 0 3 2 2 1 7 0 2 2 4 7 3 6 3 4 2 5\\
    0 0 0 0 0 0 0 0 0 0 1 0 0 0 0 0 2 7 6 2 5 7 0 2 3 4 7 3 6 3 4 2\\
    0 0 0 0 0 0 0 0 0 0 0 1 0 0 0 0 2 3 5 0 2 6 3 2 5 2 5 4 6 3 4 7\\
    0 0 0 0 0 0 0 0 0 0 0 0 1 0 0 0 0 6 6 3 2 2 1 7 4 6 3 4 7 3 6 3\\
    0 0 0 0 0 0 0 0 0 0 0 0 0 1 0 0 1 0 6 6 3 2 2 1 5 4 6 3 4 7 3 6\\
    0 0 0 0 0 0 0 0 0 0 0 0 0 0 1 0 3 5 4 6 2 3 2 2 2 5 4 6 3 4 7 3\\
    0 0 0 0 0 0 0 0 0 0 0 0 0 0 0 2 0 0 2 4 0 2 6 0 4 6 0 6 6 4 6 0\\
    0 0 0 0 0 0 0 0 0 0 0 0 0 0 0 0 4 4 4 0 4 0 0 0 0 0 0 0 0 0 0 0\\
\end{array}\right),\
G_{10}=\left(\begin{array}{c}
	 1 0 0 0 0 0 0 0 0 0 0 0 0 0 0 0 3 4 3 6 0 5 2 1 2 5 2 2 0 7 1 2\\
    0 1 0 0 0 0 0 0 0 0 0 0 0 0 0 1 1 6 5 3 4 2 0 4 3 0 2 6 0 3 3 0\\
    0 0 1 0 0 0 0 0 0 0 0 0 0 0 0 1 0 2 0 4 5 4 7 3 0 0 3 1 0 5 4 6\\
    0 0 0 1 0 0 0 0 0 0 0 0 0 0 0 0 3 6 3 7 4 3 6 4 1 7 6 6 5 2 2 0\\
    0 0 0 0 1 0 0 0 0 0 0 0 0 0 0 1 2 2 7 3 5 2 6 0 5 3 4 2 4 0 6 1\\
    0 0 0 0 0 1 0 0 0 0 0 0 0 0 0 0 2 4 3 6 3 7 4 3 6 0 1 7 6 6 5 2\\
    0 0 0 0 0 0 1 0 0 0 0 0 0 0 0 1 3 1 5 3 4 1 2 6 3 0 5 5 5 1 2 4\\
    0 0 0 0 0 0 0 1 0 0 0 0 0 0 0 1 2 4 3 4 1 4 6 1 0 0 3 4 7 2 2 5\\
    0 0 0 0 0 0 0 0 1 0 0 0 0 0 0 0 2 2 1 7 0 2 2 5 7 3 6 3 4 2 5 4\\
    0 0 0 0 0 0 0 0 0 1 0 0 0 0 0 0 3 2 2 1 7 0 2 2 4 7 3 6 3 4 2 5\\
    0 0 0 0 0 0 0 0 0 0 1 0 0 0 0 0 2 3 2 2 5 7 0 6 7 4 7 3 6 3 4 2\\
    0 0 0 0 0 0 0 0 0 0 0 1 0 0 0 0 2 6 3 2 6 1 7 4 2 3 4 7 3 6 3 4\\
    0 0 0 0 0 0 0 0 0 0 0 0 1 0 0 1 2 5 7 3 4 0 4 5 5 0 0 0 5 6 2 2\\
    0 0 0 0 0 0 0 0 0 0 0 0 0 1 0 0 1 0 6 6 3 2 2 1 5 4 6 3 4 7 3 6\\
    0 0 0 0 0 0 0 0 0 0 0 0 0 0 1 1 1 0 1 6 0 1 5 0 3 7 1 2 1 7 3 2\\
    0 0 0 0 0 0 0 0 0 0 0 0 0 0 0 2 0 6 2 0 0 4 6 0 6 4 2 0 4 6 0 6\\
    0 0 0 0 0 0 0 0 0 0 0 0 0 0 0 0 4 0 0 0 4 0 0 4 4 0 0 0 0 0 0 0\\
\end{array}\right).$$
}
So, we obatined at least $10$ new inequivalent extremal Type II $\mathbb{Z}_8$-codes of length $32.$

\end{proof}

Proposition \ref{newZ8}, together with results from  \cite{6810} and \cite{Hnovi}, yields the following theorem.

\begin{thm} 
There are at least $16$ inequivalent extremal Type II $\mathbb{Z}_8$-codes of length $32.$
\end{thm}

\begin{remark}
According to \cite{codetables}, binary $[32,15,8]$ codes are the optimal binary $[32,15]$ codes. Therefore, all constructed extremal Type II $\mathbb{Z}_8$-codes of length $32$ have optimal binary residue codes.
\end{remark}

\begin{center}{\bf Acknowledgement}\end{center}
This work has been supported by Croatian Science Foundation under the project 4571 and by the University of Rijeka under the Project uniri-iskusni-prirod-23-62-3007.

\end{document}